\documentclass[journal]{IEEEtran}
\usepackage{amsmath,amsfonts,amssymb,amsthm}
\usepackage{algorithm}
\usepackage{algorithmic}
\usepackage{array}
\usepackage{booktabs}
\usepackage{multirow}
\usepackage{cite}
\usepackage[caption=false,font=footnotesize,labelfont=rm,textfont=rm]{subfig}
\usepackage{tabularx}
\usepackage{textcomp}
\usepackage{stfloats}
\usepackage{placeins}
\usepackage{url}
\usepackage{verbatim}
\usepackage{graphicx}
\usepackage{xcolor}
\usepackage[percent]{overpic}
\def\BibTeX{{\rm B\kern-.05em{\sc i\kern-.025em b}\kern-.08em
    T\kern-.1667em\lower.7ex\hbox{E}\kern-.125emX}}

\newtheorem{proposition}{Proposition}
\newtheorem{corollary}{Corollary}
\newtheorem{lemma}{Lemma}

\begin{document}

\title{Integrated Sensing and Communication via Channel Knowledge Maps Built from Beam Feedback Coefficients}

\author{Chongda Huang,~Yue Xiao,~\IEEEmembership{Member,~IEEE,} 
        Gang Wu,~\IEEEmembership{Senior Member,~IEEE,}
        Vladimir Poulkov,~\IEEEmembership{Senior Member,~IEEE,}
        Xianfu Lei,~\IEEEmembership{Member,~IEEE}
        and Wei Xiang,~\IEEEmembership{Senior Member,~IEEE}%

\thanks{The collaboration was confined to fundamental academic research intended for open scholarly publication and dissemination. Wei Xiang’s contribution was not supported by the grant and was limited to conceptual framing, academic consultation, and manuscript review. No non-public datasets, infrastructure access, proprietary source code, intellectual property, controlled technology, or ARC-funded resources were transferred or made available through this collaboration.}

\thanks{C. Huang, Y. Xiao, and G. Wu are with the National Key Laboratory of Wireless Communications, University of Electronic Science and Technology of China, Chengdu 611731, China (e-mail: chongdahuang@std.uestc.edu.cn, xiaoyue@uestc.edu.cn).}%

\thanks{Vladimir Poulkov is with the Faculty of Telecommunications, Technical
University of Sofia, 1000 Sofia, Bulgaria (e-mail: vkp@tu-sofia.bg).}%

\thanks{X. Lei is with the School of Information Science and Technology, Southwest Jiaotong University, Chengdu 610031, China (e-mail: xflei@swjtu.edu.cn).}%

\thanks{Wei Xiang is with the School of Computing, Engineering and Mathematical
Sciences, La Trobe University, Melbourne, VIC 3086, Australia (e-mail:
w.xiang@latrobe.edu.au).}%
}

\maketitle

\begin{abstract}
Channel knowledge maps (CKMs) support wireless localization by linking channel features to user locations. Detailed channel fingerprints, however, can require substantial storage. This paper proposes a CKM-based localization framework that uses beam feedback coefficients collected during communication. These coefficients serve as compact location fingerprints within the beam management process, without requiring full channel feedback or additional sensing measurements. Using a coarse position prior, the base station searches nearby CKM entries and refines the location estimate through local interpolation. We analyze the relationship between spatial fingerprint variation and localization accuracy, and evaluate the effect of feedback quantization. Simulations show favorable storage--accuracy trade-offs relative to dense and sparsified spectral fingerprinting in the tested multipath environments.
\end{abstract}


\section{Introduction}
\IEEEPARstart{C}{ommunication} measurements reflect propagation geometry and can support localization and environment-aware services \cite{ref37,ref38}. Integrated sensing and communication (ISAC) reuses communication infrastructure, waveforms, and processing for sensing \cite{ref30,ref28}, extending joint communication--radar formulations \cite{ref27,ref32} to mobile networks \cite{ref29}. Here, compact beam-management feedback supplies the channel information for user localization.

Channel knowledge maps (CKMs) associate locations with environment-dependent channel descriptors for environment-aware operation \cite{ref12}; their construction and applications are reviewed in \cite{ref4}. Statistical map construction has been studied using expectation--maximization \cite{ref13}, while Long et al. use a line-of-sight (LoS) map to select anchors by minimizing the Bayesian Cram\'{e}r--Rao lower bound of positioning error \cite{ref14}. Channel state information (CSI) fingerprinting includes deep-learning formulations \cite{ref6}, such as DeepFi \cite{ref20} and convolutional fingerprint processing in ConFi \cite{ref21}. These methods use detailed channel measurements, which may be unavailable at the localization processor. Their storage requirements grow with the number of reference locations, antennas, and frequency samples. We instead use compact communication-side coefficients as CKM entries.

Multipath-rich and non-line-of-sight (NLOS) propagation further motivates signature-based localization \cite{ref36}. Geometric methods infer position from estimated directions, delays, and their association with physical paths \cite{ref16,ref22,ref17}. Multiple signal classification (MUSIC) provides high-resolution direction estimation under an appropriate signal and array model \cite{ref39}; spatial smoothing mitigates rank loss caused by coherent paths \cite{ref40}. Neither operation alone resolves the association between a reflected path and the physical user. Fingerprinting uses the aggregate propagation signature to distinguish locations without identifying each path \cite{ref7,ref8}. Hu and Chen developed angle delay power spectrum (ADPS) fingerprints and adjustable-weight weighted $K$-nearest neighbor (AWKNN) matching for orthogonal frequency-division multiplexing (OFDM) localization \cite{hu2024adps}. ADPS-AWKNN exploits angle-delay structure and requires full channel measurements to construct its spectral fingerprints.

Communication feedback offers another way to use this structure. New Radio (NR) CSI reporting describes spatial and frequency-dependent channel information \cite{ref1,ref15}, while beam management maintains directional communication links \cite{ref19}. Precoder-feedback overhearing has been shown to reveal location information \cite{ref5}, and adaptive CSI feedback has been investigated for ISAC operation \cite{ref34}. Commercial NR CSI has also been used for indoor positioning \cite{ref24,ref25}. The remaining question is whether compact reports retain enough spatial information for accurate localization, and how coefficient precision and map sampling affect the result.

Table~\ref{tab:localization_interfaces} compares the localization inputs and fingerprint representations of the selected methods. Both ADPS fingerprinting and the proposed method use aggregate multipath signatures without explicit path separation. The proposed CKM directly stores and matches the selected beam-feedback indices and complex coefficients.

\begin{table*}[!t]
\centering
\caption{Comparison of Localization Inputs and Fingerprint Representations}
\label{tab:localization_interfaces}
\small
\renewcommand{\arraystretch}{1.25}
\setlength{\tabcolsep}{5pt}
\begin{tabularx}{\textwidth}{@{}>{\raggedright\arraybackslash}p{0.24\textwidth}
*{4}{>{\centering\arraybackslash}X}@{}}
\toprule
Method
& Direct use of communication beam feedback
& No full CSI required for localization feature generation
& No explicit path separation or association
& Sparse index--value fingerprint\\
\midrule
Geometry-based localization~\cite{ref16}
& --- & $\triangle$ & --- & N/A\\
Dense ADPS-AWKNN~\cite{hu2024adps}
& --- & --- & $\checkmark$ & ---\\
Top-\(M\) sparse ADPS-AWKNN (extension)
& --- & --- & $\checkmark$ & $\checkmark$\\
Proposed feedback-coefficient CKM
& $\checkmark$ & $\checkmark$ & $\checkmark$ & $\checkmark$\\
\bottomrule
\end{tabularx}

\vspace{3pt}
\begin{minipage}{\textwidth}
\footnotesize
$\checkmark$: the stated property holds in the compared formulation;
---: the stated property is absent in that formulation;
$\triangle$: depends on the measurement and parameter-extraction configuration;
N/A: no fingerprint map is used.
The full-CSI column covers localization feature generation, including path-parameter extraction for the geometric method and query-spectrum construction for both ADPS variants. The Top-\(M\) extension is evaluated in Section~V.
For the proposed method, features are formed from already available beam-feedback indices and coefficients; communication-side feedback generation follows the assumptions in Section~II.
\end{minipage}
\end{table*}

The network selects angle-delay ports using communication-side channel knowledge, and the user equipment (UE) estimates and reports their coefficients. The network builds the CKM from location-labeled reports and matches new reports for localization. The reciprocity-aided model in Section~II uses the same measurement and feedback chain for communication adaptation and localization. The number of retained ports controls the map size.

A coarse location prior limits the search to a nearby reference region. The estimator ranks candidates by a phase-invariant coefficient distance and uses regularized nonnegative least squares (NNLS) for local interpolation, with explicit fallback rules. The prior is supplied externally, for example by a network positioning service or a previous position estimate. This local formulation addresses feature ambiguity without assuming globally unique fingerprints. The analysis traces perturbations from the channel to the coefficients and identifies the smoothness and local identifiability conditions for stable position estimation.

The main contributions are:
\begin{itemize}
\item A compact CKM is constructed from selected angle-delay ports and complex feedback coefficients. A port-measurement model connects communication-side observations to these entries and allows localization to use the reports directly.
\item Coarse-prior-guided matching and interpolation exploit the retained spatial information. A phase-projected identifiability condition explains local stability, while finite-map and quantization bounds identify conditions that preserve nearby matches.
\item We compare the proposed method's storage--accuracy trade-off with dense and Top-\(M\) sparse ADPS-AWKNN in two scenes. The proposed method yields lower localization root-mean-square error (RMSE) at similar storage budgets and uses less storage in tested configurations with similar accuracy. We also evaluate map density, port budget, quantization, and coefficient estimation error.
\end{itemize}

Section~II describes the communication-side framework and feedback generation. Section~III specifies the CKM and localization estimator. Section~IV develops the local geometry and deterministic error characterization. Section~V presents the numerical results, and Section~VI concludes the paper.

\section{Beam-Feedback-Based ISAC Framework}
\label{sec:framework}
Communication-side channel estimation supplies the coefficients reported by the UE to the base station (BS). Fig.~\ref{fig:ckm_diagram} shows the beam-management and feedback process, including channel state information reference signal (CSI-RS) configuration and transmission, UE channel estimation, outer-layer beam description and inner-layer coefficient extraction, and uplink reporting.

\begin{figure*}[!t]
\centering
\includegraphics[width=\textwidth]{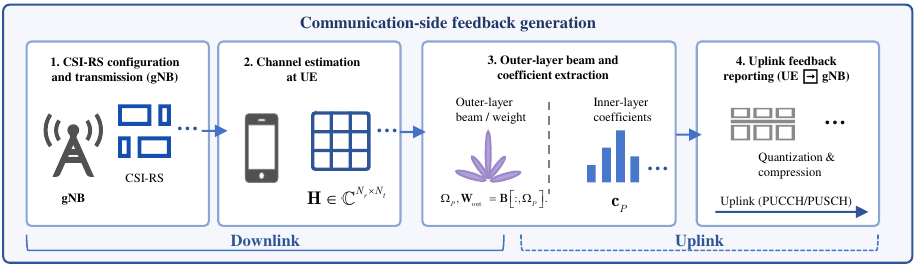}
\caption{Communication-side beam management and feedback generation.}
\label{fig:ckm_diagram}
\end{figure*}

The communication model is motivated by the Further enhanced Type II port selection codebook in Release-17 Technical Specification (TS) 38.214, clause 5.2.2.2.7 \cite{refR17}. That specification defines port-selection reporting; the joint discrete Fourier transform (DFT) atoms and scalar quantizer below are our numerical realization. The BS is assumed to configure ports using calibrated uplink-derived support knowledge. In the experiments, the available downlink channel map supplies that support and the receive combiner directly. Uplink-to-downlink support transfer, calibration, and combiner acquisition are therefore idealized. The main comparisons use ideal floating-point coefficients, and Section~\ref{subsec:quantization} evaluates their finite-precision quantization. A separate experiment in Fig.~\ref{fig:coefficient_estimation} adds coefficient estimation error while keeping the support and combiner fixed.

Let \(\mathbf p\in\mathcal A\subset\mathbb R^2\) denote the UE position in a multiple-input multiple-output (MIMO) OFDM system. At frequency sample \(k\), the channel is \(\mathbf H_k(\mathbf p)\in\mathbb C^{N_r\times N_t}\), with
\begin{equation}
\begin{split}
\mathbf H_k(\mathbf p)
={}&\sum_{\ell=1}^{L}\beta_\ell(\mathbf p)
 e^{-j2\pi f_k\tau_\ell(\mathbf p)}\\
&\quad\cdot\mathbf a_r(\boldsymbol\phi_\ell^r(\mathbf p),k)
 \mathbf a_t^H(\boldsymbol\phi_\ell^t(\mathbf p),k).
\end{split}
\label{eq:Hk}
\end{equation}

Here \(\beta_\ell\), \(\tau_\ell\), and \(\boldsymbol\phi_\ell^{r,t}\) are the complex path gain, delay, and arrival/departure angle vectors. We write \(\mathbf H(\mathbf p)=\{\mathbf H_k(\mathbf p)\}_{k=1}^{K}\). In the channel synthesis, \(f_k=(k-1)\Delta f\) is a baseband frequency offset; the carrier is \(f_c\). The gain is \(\beta_\ell=\sqrt{10^{P_\ell^{\rm dB}/10}/K}\exp(j\psi_\ell)\), where \(P_\ell^{\rm dB}\) and \(\psi_\ell\) are the supplied path power and phase. The synthesis convention uses \(\mathbf a_r\mathbf a_{t,\mathrm{syn}}^T\), so the transmit response in \eqref{eq:Hk} is defined by \(\mathbf a_t=\mathbf a_{t,\mathrm{syn}}^*\). The supplied path phase already includes the carrier contribution.

The receive covariance and unit-norm dominant combining vector are:
\begin{equation}
\mathbf R_r=\frac1K\sum_{k=1}^{K}\mathbf H_k\mathbf H_k^H,
\qquad
\mathbf u=\operatorname{eig}_{\max}(\mathbf R_r)\in\mathbb C^{N_r}.
\label{eq:receive_combining}
\end{equation}

The receive-combined effective channel is
\begin{equation}
\mathbf X=
[\mathbf H_1^H\mathbf u,\ldots,\mathbf H_K^H\mathbf u]
\in\mathbb C^{N_t\times K}.
\label{eq:effective_channel}
\end{equation}

The dependence on \(\mathbf p\) is omitted where unambiguous. This operation retains the dominant receive mode rather than the complete multi-receive-antenna channel.

For an \(N_x\times N_y\) transmit uniform planar array (UPA), \(N_t=N_xN_y\), define the unitary DFT matrix
\begin{equation}
[\mathbf F_n]_{a,b}=n^{-1/2}e^{-j2\pi ab/n},
\quad a,b=0,\ldots,n-1,
\label{eq:dft}
\end{equation}
and use the antenna ordering specified below
\begin{equation}
\mathbf W_\theta=\mathbf F_{N_x}\otimes\mathbf F_{N_y},
\qquad
\mathbf G=\mathbf W_\theta^H\mathbf X\mathbf F_K^*
\in\mathbb C^{N_t\times K}.
\label{eq:angle_delay_transform}
\end{equation}

The symbol \(*\) denotes elementwise conjugation. The spatial dictionary and channel vectorization must use the same antenna ordering. This representation follows the beamspace view of multiantenna channels \cite{sayeed2002deconstructing}, whose sparse signatures have also been studied for localization \cite{ref18}. The frequency transform represents sampled spectral structure; its bins need not resolve individual physical paths.

Let \(D=N_tK\) and \(\mathbf g=\operatorname{vec}(\mathbf G)\in\mathbb C^D\), using column-major vectorization. For \(P\le D\), retain the \(P\) largest values of \(|g_j|^2\). Equal energies are ordered by their original linear index, and the retained indices are subsequently sorted in ascending order:
\begin{equation}
\Omega_P=\{j_1<\cdots<j_P\},\qquad
\mathbf c_P=[g_{j_1},\ldots,g_{j_P}]^T.
\label{eq:port_selection}
\end{equation}

The corresponding space-frequency atoms are columns of
\(\mathbf B=\mathbf F_K\otimes\mathbf W_\theta\in\mathbb C^{D\times D}\).
Hence the outer descriptor can be represented without storing a dense matrix:
\begin{equation}
\begin{aligned}
\mathbf W_{\mathrm{out}}&=\mathbf B_{:,\Omega_P}\in\mathbb C^{N_tK\times P},\\
\mathbf V_{\mathrm{sb}}&=\mathbf c_P,\qquad
\mathbf F=(\Omega_P,\mathbf c_P).
\end{aligned}
\label{eq:reported_feature}
\end{equation}

This is a joint spatial--frequency outer representation. With spatial index \(s\in\{1,\ldots,N_t\}\) and frequency-DFT index \(d\in\{1,\ldots,K\}\), port \(j=s+N_t(d-1)\) follows column-major vectorization. The corresponding one-based DFT entry is \([\mathbf F_K]_{k,d}=K^{-1/2}e^{-j2\pi(k-1)(d-1)/K}\). Its frequency-\(k\) spatial weight is
\begin{equation}
\mathbf b_{j,k}=[\mathbf F_K]_{k,d}[\mathbf W_\theta]_{:,s}.
\label{eq:port_weight}
\end{equation}

Thus \(\mathbf W_{\mathrm{out}}\) stacks \(K\) spatial-weight blocks, rather than representing a single narrowband beam matrix. The transform conventions give
\begin{equation}
\operatorname{vec}(\mathbf X)=\mathbf B\mathbf g,
\qquad \operatorname{vec}(\hat{\mathbf X})=\mathbf W_{\mathrm{out}}\mathbf c_P.
\label{eq:dictionary_reconstruction}
\end{equation}

All map entries share the same dictionary; \(\Omega_P\) identifies the selected outer atoms and \(\mathbf c_P\) gives their weights. The support specifies the retained spatial-frequency coordinates, while the coefficients preserve their relative amplitude and phase. Explicit CSI-feedback designs involve additional representation and reporting choices \cite{ref26}; floating-point coefficients serve as the baseline for the finite-precision experiment in Section~\ref{subsec:quantization}.

To connect these coefficients to communication measurements, suppose the BS configures the selected atoms and assigns each port a separable pilot resource. After pilot despreading, the UE observes
\begin{equation}
\mathbf y_{j,k}=\sqrt{E_j}\mathbf H_k\mathbf b_{j,k}+\mathbf n_{j,k},
\qquad
t_j=\sum_{k=1}^{K}\mathbf u^H\mathbf y_{j,k},
\label{eq:pilot_observation}
\end{equation}
where \(E_j>0\) is the known pilot-energy factor. The unitary dictionaries imply \(\sum_k\|\mathbf b_{j,k}\|^2=1\). Orthogonal pilot codes or separate resources are assumed so that other ports have been removed by despreading. The same unit-norm combiner is used across the \(K\) samples. Since \(\mathbf X_{:,k}=\mathbf H_k^H\mathbf u\),
\begin{equation}
\begin{aligned}
t_j&=\sqrt{E_j}\sum_k[\mathbf F_K]_{k,d}
 \mathbf u^H\mathbf H_k[\mathbf W_\theta]_{:,s}+\nu_j\\
&=\sqrt{E_j}\,g_j^*+\nu_j,\qquad
\widetilde c_j=t_j^*/\sqrt{E_j}.
\end{aligned}
\label{eq:pilot_coefficient}
\end{equation}

Each coefficient is obtained by conjugating and energy-normalizing its observation, and is reported with its port index in the shared dictionary. The simulation supplies the full-channel combiner; practical acquisition requires receive-side training. Equation~\eqref{eq:pilot_observation} models the observation within the NR CSI-RS architecture \cite{ref2,ref23}, without specifying resource allocation.

For independent \(\mathbf n_{j,k}\sim\mathcal{CN}(\mathbf0,\sigma_n^2\mathbf I)\) and fixed \(\|\mathbf u\|_2=1\), the sum \(\nu_j=\sum_k\mathbf u^H\mathbf n_{j,k}\) has variance \(K\sigma_n^2\). Thus \eqref{eq:pilot_coefficient} gives
\begin{equation}
\widetilde c_j=c_j+\frac{\nu_j^*}{\sqrt{E_j}},\qquad
\mathbb E|\widetilde c_j-c_j|^2=\frac{K\sigma_n^2}{E_j}.
\label{eq:coefficient_observation_variance}
\end{equation}
The factor \(K\) follows from the unaveraged sum in \eqref{eq:pilot_observation}; DFT normalization is included in the transmitted atom. Correct uplink delivery transmits these estimates with their UE-side estimation error.

The amplitude--phase quantizer is applied after port selection. Here \(c_j=g_j\) is the ideal coefficient, \(\widetilde c_j\) is its measurement estimate, and \(\widehat c_j\) is the reconstructed quantized coefficient. The main comparisons and quantization experiment use \(\widetilde c_j=c_j\); the separate estimation-error experiment uses noisy coefficients without quantization. For one report, with \(j\) ranging over its selected support, define
\begin{equation}
A=\max_j|\widetilde c_j|,\quad r_j=|\widetilde c_j|/A,\quad
\phi_j=\operatorname{mod}(\arg\widetilde c_j,2\pi).
\label{eq:amplitude_phase_normalization}
\end{equation}

If \(A=0\), set every \(r_j=0\); for \(\widetilde c_j=0\), set \(\phi_j=0\). With amplitude and phase budgets \(b_a,b_\phi\), let \(L_a=2^{b_a}-1\) and \(L_\phi=2^{b_\phi}\). The quantization and reconstruction rules are
\begin{equation}
\begin{aligned}
q_{a,j}&=\min\{L_a,\max\{0,\lfloor L_a r_j+1/2\rfloor\}\},\\
q_{\phi,j}&=\operatorname{mod}\!\left(
 \left\lfloor\frac{L_\phi\phi_j}{2\pi}+\frac12\right\rfloor,L_\phi\right),\\
\hat c_j&=A\frac{q_{a,j}}{L_a}
 \exp\!\left(j\frac{2\pi q_{\phi,j}}{L_\phi}\right).
\end{aligned}
\label{eq:amplitude_phase_quantizer}
\end{equation}

Amplitude levels include both zero and one. Half-step ties round toward the increasing code; phase codes then wrap modulo \(L_\phi\). A component left unquantized retains its float64 normalized amplitude or phase. All selected indices are retained, including those whose reconstructed amplitude is zero. The report contains the index array, float64 scale \(A\), quantization indices, and bit allocations. Payload accounting uses 11 bits per index for \(D=1664\) and 64 bits for \(A\). Joint quantization uses \(b_a+b_\phi\) bits per coefficient, excluding these fields. This scalar quantizer is separate from the standardized codebook quantizers.

\section{Feedback-Coefficient CKM and Localization}
\label{sec:localization}
The CKM associates reference locations with compact feedback reports. Localization combines coefficient matching with local interpolation, as summarized in Algorithm~\ref{alg:localization}. Fig.~\ref{fig:ckm_framework} illustrates map construction and the localization procedure.

\begin{figure*}[!t]
\centering
\includegraphics[width=0.88\textwidth]{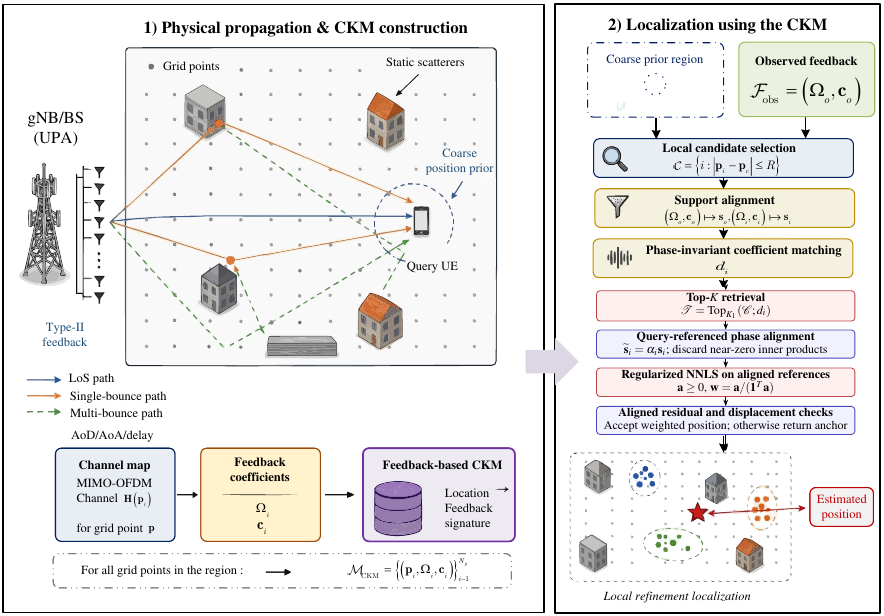}
\caption{Feedback-coefficient CKM construction and query-aligned localization.}
\label{fig:ckm_framework}
\end{figure*}

The online input consists of the reported support, coefficient values, and coarse position prior. In this section and Fig.~\ref{fig:ckm_framework}, \(\mathbf c_i\) denotes the coefficient vector in database entry \(i\): it contains \(\widetilde c_j\) in the floating-point case or \(\widehat c_j\) after reconstruction. These database-indexed vectors are distinct from the ideal port coefficient \(c_j=g_j\). The spatial map and its sampled realization are
\begin{equation}
\mathcal M:\mathbf p\mapsto\mathbf F(\mathbf p),\qquad
\mathcal M_{\mathrm{CKM}}=\{(\mathbf p_i,\Omega_i,\mathbf c_i)\}_{i=1}^{N_g}.
\label{eq:ckm_database}
\end{equation}

The position-indexed map may be discontinuous where the selected support changes. Define a sparse embedding \(\mathbf s(\mathbf p)\in\mathbb C^D\) by placing the retained coefficients at \(\Omega_P(\mathbf p)\) and zeros elsewhere. Only the selected indices and coefficients are stored; the sparse embedding supplies common coordinates for matching.

For a query \(\mathbf F^{\mathrm{obs}}\), a coarse prior \(\mathbf p_c\), and search radius \(R\), the candidate set is
\begin{equation}
\mathcal C=\{i:\|\mathbf p_i-\mathbf p_c\|\le R\}.
\label{eq:candidate_set}
\end{equation}

During leave-one-out evaluation, the query reference and identified raw-channel duplicates are excluded. In an operational query, this exclusion is unnecessary unless the query itself has been inserted into the reference database.

For two nonzero sparse embeddings \(\mathbf s_o\) and \(\mathbf s_i\), matching uses
\begin{equation}
d_i=
\frac{\sqrt{\max\{\|\mathbf s_o\|^2+\|\mathbf s_i\|^2
                  -2|\mathbf s_o^H\mathbf s_i|,0\}}}
{\max\{(\|\mathbf s_o\|+\|\mathbf s_i\|)/2,\epsilon_{\mathrm{mach}}\}}.
\label{eq:feature_distance}
\end{equation}

This is the minimum Euclidean residual over a common phase rotation, normalized by the mean feature norm. The inner product is computed on the intersection of supports, while coefficients on nonoverlapping ports contribute to the residual. Equation~\eqref{eq:feature_distance} therefore jointly measures support mismatch and coefficient differences, including candidates with disjoint supports.

Phase minimization removes the arbitrary phase of the dominant receive eigenvector during ranking while preserving relative coefficient magnitudes and differences in feature norms. Candidates are sorted by increasing \(d_i\), with database order resolving ties. With retrieval size \(K_{\mathrm r}=8\) and \(K_1=\min\{K_{\mathrm r},|\mathcal C|\}\), the selected set is
\begin{equation}
\mathcal T=\operatorname{Top}_{K_1}(\mathcal C;d_i).
\label{eq:topk_set}
\end{equation}

The first retrieved point is the anchor \(\mathbf p_a\). Retain \(\mathcal T_+=\{i\in\mathcal T:|\mathbf s_i^H\mathbf s_o|>10^{-12}\|\mathbf s_i\|\|\mathbf s_o\|\}\), excluding zero-norm vectors, and align each retained reference to the query:
\begin{equation}
\alpha_i=\frac{\mathbf s_i^H\mathbf s_o}{|\mathbf s_i^H\mathbf s_o|},\qquad
\widetilde{\mathbf s}_i=\alpha_i\mathbf s_i,\quad i\in\mathcal T_+.
\label{eq:query_phase_alignment}
\end{equation}

Interpolation uses \(\widetilde{\mathbf z}_i=[\operatorname{Re}\widetilde{\mathbf s}_i^T,\operatorname{Im}\widetilde{\mathbf s}_i^T]^T\) and \(\widetilde{\mathbf Z}=[\widetilde{\mathbf z}_i]_{i\in\mathcal T_+}\). The query vector is represented as \(\mathbf z_o=[\operatorname{Re}\mathbf s_o^T,\operatorname{Im}\mathbf s_o^T]^T\).
We obtain the weights by regularizing the nonnegative least-squares formulation in \cite{lawson1995least}:
\begin{equation}
\mathbf a=\arg\min_{\mathbf a\ge0}
 \|\widetilde{\mathbf Z}\mathbf a-\mathbf z_o\|_2^2+\lambda\|\mathbf a\|_2^2,
\qquad \lambda=10^{-3}.
\label{eq:nnls}
\end{equation}

When \(\mathbf1^T\mathbf a>0\), the normalized weights, interpolated location, and reconstruction residual are
\begin{equation}
\begin{aligned}
\mathbf w&=\mathbf a/(\mathbf1^T\mathbf a),&
\mathbf p_{\mathrm{int}}&=\sum_{i\in\mathcal T_+}w_i\mathbf p_i,\\
r_{\mathrm{int}}&=
\frac{\|\mathbf z_o-\widetilde{\mathbf Z}\mathbf w\|_2}
{\max\{\|\mathbf z_o\|_2,\epsilon_{\mathrm{mach}}\}}.
\end{aligned}
\label{eq:interpolation}
\end{equation}

NNLS and acceptance use the same aligned matrix, without amplitude normalization. Independent rotations \(\mathbf s_i\mapsto e^{j\theta_i}\mathbf s_i\) and \(\mathbf s_o\mapsto e^{j\theta_o}\mathbf s_o\) rotate every aligned reference by \(e^{j\theta_o}\). The objective and residual are unchanged, as are the unique ridge-NNLS weights and final position. Amplitude scaling by \(\kappa\), however, changes the effective regularization to \(\lambda/|\kappa|^2\).

With fewer than four retained references, the anchor is returned. Otherwise, let
\begin{equation}
h_{\mathcal T}=\max_{i\in\mathcal T_+}\|\mathbf p_i-\bar{\mathbf p}_{\mathcal T}\|,
\quad
\bar{\mathbf p}_{\mathcal T}=|\mathcal T_+|^{-1}\sum_{i\in\mathcal T_+}\mathbf p_i.
\label{eq:patch_extent}
\end{equation}

The interpolated position is accepted if \(r_{\mathrm{int}}\) is finite, \(r_{\mathrm{int}}\le0.35\), and
\(\|\mathbf p_{\mathrm{int}}-\mathbf p_a\|\le\max\{1.5h_{\mathcal T},1~\mathrm m\}\).
Otherwise the estimator returns the anchor; it also does so if \(\mathbf1^T\mathbf a\le0\). The numerical distance clips its squared residual at zero before taking the square root. Zero-norm inputs use the denominator floor: two zero vectors tie at zero distance, and a zero query with nonzero references has no directional discrimination. If the output coordinates are nonfinite or the candidate set is empty, evaluation returns \(\mathbf p_c\). 

\begin{algorithm}[!t]
\caption{Feedback-Coefficient CKM Localization}
\label{alg:localization}
\begin{algorithmic}[1]
\STATE \textbf{Offline:} For each reference \(\mathbf p_i\), compute \(\mathbf X_i\), \(\mathbf G_i\), and \((\Omega_i,\mathbf c_i)\) using \eqref{eq:receive_combining}--\eqref{eq:port_selection}.
\STATE Store \((\mathbf p_i,\Omega_i,\mathbf c_i)\).
\STATE \textbf{Input:} Query \((\Omega_o,\mathbf c_o)\), prior \(\mathbf p_c\), radius \(R\), and reference map.
\STATE Initialize failure flag \(f=0\).
\STATE Form \(\mathcal C\) using \eqref{eq:candidate_set}; apply evaluation exclusions if required.
\IF{\(\mathcal C=\varnothing\)}
\STATE Set \(f=1\); return \((\mathbf p_c,f)\).
\ENDIF
\STATE Embed supports in common coordinates and compute \(d_i\) from \eqref{eq:feature_distance}.
\STATE Retrieve \(\mathcal T=\operatorname{Top}_{K_1}(\mathcal C;d_i)\), \(K_1=\min\{8,|\mathcal C|\}\); set anchor \(\mathbf p_a\).
\STATE Initialize \(\hat{\mathbf p}=\mathbf p_a\).
\STATE Form \(\mathcal T_+\), align references to the query using \eqref{eq:query_phase_alignment}, and build \(\widetilde{\mathbf Z}\).
\IF{\(|\mathcal T_+|\ge4\)}
\STATE Solve the regularized NNLS problem in \eqref{eq:nnls}.
\IF{\(\mathbf1^T\mathbf a>0\)}
\STATE Normalize \(\mathbf w=\mathbf a/(\mathbf1^T\mathbf a)\).
\STATE Compute \(\mathbf p_{\mathrm{int}}\), \(r_{\mathrm{int}}\), and \(h_{\mathcal T}\).
\IF{residual and anchor-displacement checks pass}
\STATE Set \(\hat{\mathbf p}=\mathbf p_{\mathrm{int}}\).
\ENDIF
\ENDIF
\ENDIF
\IF{\(\hat{\mathbf p}\) has a nonfinite component}
\STATE Set \(\hat{\mathbf p}=\mathbf p_c\) and \(f=1\).
\ENDIF
\STATE \textbf{Return} \((\hat{\mathbf p},f)\); a finite anchor fallback retains \(f=0\).
\end{algorithmic}
\end{algorithm}

For two supports of size at most \(P\), the distance can be evaluated by merging sorted indices in \(O(P)\) operations. Thus, candidate ranking does not require reconstruction of the \(N_r\times N_t\times K\) channel tensor. The NNLS problem has at most eight unknowns; its data can be restricted to the union of the retrieved and query supports.

\section{Local Geometry and Error Characterization}
\label{sec:theory}
\subsection{Local Identifiability}
Localization information arises because the channel parameters in \eqref{eq:Hk} vary with position \cite{ref11}. Consider a neighborhood of \(\mathbf p_0\) without path birth, disappearance, or a visibility boundary. For the real path parameter vector
\begin{equation}
\boldsymbol\theta_\ell=
[\tau_\ell,(\boldsymbol\phi_\ell^r)^T,
(\boldsymbol\phi_\ell^t)^T,\operatorname{Re}\beta_\ell,
\operatorname{Im}\beta_\ell]^T,
\end{equation}
continuous differentiability gives
\begin{equation}
\boldsymbol\theta_\ell(\mathbf p_0+\Delta\mathbf p)
=\boldsymbol\theta_\ell(\mathbf p_0)
+\mathbf J_{\theta_\ell}\Delta\mathbf p+o(\|\Delta\mathbf p\|).
\label{eq:path_expansion}
\end{equation}

Applying the chain rule to the array responses and path phases yields
\begin{equation}
\Delta\mathbf H_k
=\mathcal J_{H_k}[\Delta\mathbf p]+o(\|\Delta\mathbf p\|),
\qquad
\mathcal J_{H_k}=\sum_\ell
D_{\theta_\ell}\mathbf H_{k,\ell}\mathbf J_{\theta_\ell}.
\label{eq:channel_derivative}
\end{equation}

Here \(\mathcal J_{H_k}\) is a real-linear operator from \(\mathbb R^2\) to channel matrices.

The receive covariance differential is
\begin{equation}
D\mathbf R_r[\mathbf v]=\frac1K\sum_k
\big(D\mathbf H_k[\mathbf v]\mathbf H_k^H
+\mathbf H_kD\mathbf H_k[\mathbf v]^H\big).
\label{eq:receive_derivative}
\end{equation}

Suppose its largest eigenvalue is simple with a locally positive eigengap. A smooth phase convention can then be selected for \(\mathbf u\) locally. Subspace perturbation bounds control changes in its projector by changes in \(\mathbf R_r\) \cite{ref9,ref10}. Both the channel and the combiner contribute to the effective-channel differential:
\begin{equation}
D\mathbf X_{:,k}[\mathbf v]
=D\mathbf H_k[\mathbf v]^H\mathbf u+
\mathbf H_k^HD\mathbf u[\mathbf v].
\label{eq:effective_derivative}
\end{equation}

The second term accounts for the change in receive combining as the UE moves.

The dictionary is fixed, so
\begin{equation}
D\mathbf G[\mathbf v]
=\mathbf W_\theta^HD\mathbf X[\mathbf v]\mathbf F_K^*.
\label{eq:coefficient_derivative}
\end{equation}

If the \(P\)th and \((P+1)\)th largest port energies are separated by a positive gap, continuity preserves the selected support in a sufficiently small neighborhood. On this neighborhood, selection is a fixed linear operator. For \(P=D\), no support condition is needed. A smooth local phase representative of the selected coefficients therefore admits
\begin{equation}
\bar\Phi(\mathbf p_0+\Delta\mathbf p)
=\bar\Phi(\mathbf p_0)+\mathbf J_{\bar\Phi}(\mathbf p_0)
 \Delta\mathbf p+o(\|\Delta\mathbf p\|).
\label{eq:feature_expansion}
\end{equation}

Derivatives and singular values are evaluated using the real augmentation of complex features. The sparse report \(\Phi\) can switch coordinates across support boundaries; the smooth precursor \(\bar\Phi\) is defined only within the local chart. The resulting chain is
\begin{equation}
\begin{aligned}
\mathbf p&\longmapsto\mathbf H\longmapsto(\mathbf u,\mathbf X)
\longmapsto\mathbf G\\
&\longmapsto(\Omega_P,\mathbf c_P)
\longmapsto\text{coefficient matching}.
\end{aligned}
\label{eq:feature_chain}
\end{equation}

\begin{proposition}[Local continuity]
\label{prop:lipschitz}
Suppose the path and steering parameters are continuously differentiable, the dominant receive eigenvalue is separated, and the selected support is locally unchanged. On a sufficiently small neighborhood, a smooth phase representative \(\bar\Phi\) satisfies
\begin{equation}
\|\bar\Phi(\mathbf p_1)-\bar\Phi(\mathbf p_2)\|
\le L_\Phi\|\mathbf p_1-\mathbf p_2\|.
\label{eq:phi_lipschitz}
\end{equation}
\end{proposition}
\begin{proof}
Equations~\eqref{eq:channel_derivative}--\eqref{eq:coefficient_derivative} are compositions of continuously differentiable maps under the stated eigengap and support conditions. On a sufficiently small closed convex neighborhood their derivatives are bounded. Integrating the feature derivative along the line segment joining two points gives \eqref{eq:phi_lipschitz}.
\end{proof}

To study identifiability under the matching rule, define the phase-orbit residual
\begin{equation}
\begin{aligned}
\delta(\mathbf x,\mathbf y)&=\min_\theta\|\mathbf x-e^{j\theta}\mathbf y\|\\
&=\sqrt{\|\mathbf x\|^2+\|\mathbf y\|^2-2|\mathbf x^H\mathbf y|}.
\end{aligned}
\label{eq:phase_residual}
\end{equation}

For nonzero features, \eqref{eq:feature_distance} divides this residual by their mean norm. It removes common phase while retaining amplitude differences. Let \(\mathbf c(\mathbf p)=\bar\Phi(\mathbf p)\), embedded on the fixed support, and \(\mathbf c_0=\mathbf c(\mathbf p_0)\ne0\). The real-linear projection away from the common-phase tangent is
\begin{equation}
\Pi_{\mathbf c_0}(\mathbf v)=\mathbf v-j\mathbf c_0
 \frac{\operatorname{Re}\{(j\mathbf c_0)^H\mathbf v\}}{\|\mathbf c_0\|^2}.
\label{eq:phase_projection}
\end{equation}

This projection preserves the amplitude direction \(\mathbf c_0\).

\begin{lemma}[Local phase-invariant identifiability]
\label{lem:phase_identifiability}
Suppose the smooth, unquantized, fixed-support feature in Proposition~\ref{prop:lipschitz} has \(\|\mathbf c_0\|>\epsilon_{\mathrm{mach}}\) and satisfies
\begin{equation}
\sigma_{\min}\!\left(\Pi_{\mathbf c_0}D\mathbf c(\mathbf p_0)\right)>0,
\label{eq:identifiability}
\end{equation}
where the operator is represented in real coordinates. Then, on a sufficiently small convex neighborhood, the normalized phase-invariant distance \(d_0\) in \eqref{eq:feature_distance} satisfies
\begin{equation}
m\|\mathbf p_1-\mathbf p_2\|\le d_0(\mathbf p_1,\mathbf p_2)
\le L\|\mathbf p_1-\mathbf p_2\|,
\label{eq:phi_bilipschitz}
\end{equation}
for some \(m,L>0\), provided the feature norms exceed the numerical denominator floor.
\end{lemma}
\begin{proof}
Choose a neighborhood with \(a\le\|\mathbf c(\mathbf p)\|\le b\), \(a>\epsilon_{\mathrm{mach}}\), and \(|\mathbf c_0^H\mathbf c(\mathbf p)|\ge s>0\). For features with \(\mathbf c_0^H\mathbf x\ne0\), define the phase-fixing map and the resulting position-dependent representative by
\begin{equation}
\begin{aligned}
T_{\mathbf c_0}(\mathbf x)&=\mathbf x
 \frac{(\mathbf c_0^H\mathbf x)^*}{|\mathbf c_0^H\mathbf x|},\\
\mathbf h(\mathbf p)&=T_{\mathbf c_0}(\mathbf c(\mathbf p)).
\end{aligned}
\label{eq:phase_chart}
\end{equation}

The derivative of \(\mathbf h\) at \(\mathbf p_0\) is \(\Pi_{\mathbf c_0}D\mathbf c(\mathbf p_0)\). By continuity, shrink the neighborhood so that the derivative variation is less than half its minimum singular value \(\mu\). The integral derivative formula gives
\(m_h\|\mathbf p_1-\mathbf p_2\|\le\|\mathbf h_1-\mathbf h_2\|\le L_h\|\mathbf p_1-\mathbf p_2\|\), with \(m_h=\mu/2\).

For scalar \(u,v\) with magnitudes at least \(s\),
\(|u/|u|-v/|v||\le2|u-v|/s\).
Apply this inequality to the phase factors of \(\mathbf x\) and an optimally rotated \(\mathbf y\). Since the chart is invariant to that rotation,
\begin{equation}
\begin{aligned}
\delta(\mathbf x,\mathbf y)
&\le\|T_{\mathbf c_0}(\mathbf x)-T_{\mathbf c_0}(\mathbf y)\|\\
&\le C\delta(\mathbf x,\mathbf y),\qquad
C=1+\frac{2b\|\mathbf c_0\|}{s}.
\end{aligned}
\label{eq:chart_distance}
\end{equation}

The mean-norm denominator lies in \([a,b]\), giving \(m=m_h/(Cb)\) and \(L=L_h/a\).
\end{proof}

Condition~\eqref{eq:identifiability} requires two linearly independent position directions after removing common phase.

\subsection{Localization Error Analysis}
Fix a query \(\mathbf p^\star\) and a nonempty candidate index set \(\mathcal C\) after prior restriction and evaluation exclusions. Its position set is \(\mathcal D_c=\{\mathbf p_i:i\in\mathcal C\}\). Let \(\mathcal N_r(\mathbf p^\star)\) be the closed radius-\(r\) neighborhood within the position domain. The nearest geometric reference is
\begin{equation}
\mathbf p_{\mathrm{nn}}=\arg\min_{\mathbf p_i\in\mathcal D_c}
\|\mathbf p_i-\mathbf p^\star\|,\qquad
e_g=\|\mathbf p_{\mathrm{nn}}-\mathbf p^\star\|.
\label{eq:grid_error}
\end{equation}

For a complete square grid of spacing \(\Delta\), with the enclosing cell available, \(e_g\le\Delta/\sqrt2\). This grid bound requires the enclosing references to remain available after cropping, exclusions, and prior restriction. Also, \(e_g\) is a lower bound only for estimators constrained to return an admissible grid point; interpolated estimates can lie between references.

Let \(d_0\) be the normalized phase-invariant distance between unquantized coefficients defined in Lemma~\ref{lem:phase_identifiability}, with constants \(m,L>0\). Let \(d_q\) be the score computed from the encoded map and query, with the selected supports and port budget held fixed. On the local reference neighborhood assume
\begin{equation}
|d_q(\mathbf p^\star,\mathbf p_i)
-d_0(\mathbf p^\star,\mathbf p_i)|\le\varepsilon_d.
\label{eq:distance_distortion}
\end{equation}

The quantity \(\varepsilon_d\) bounds score distortion, not channel-reconstruction normalized mean-square error (NMSE). For exact floating-point evaluation of the same features, \(d_q=d_0\) up to arithmetic error.

\begin{proposition}[Local nearest-reference matching]
\label{prop:nearest}
Suppose \(\mathcal D_c\ne\varnothing\), \(\mathbf p_{\mathrm{nn}}\in\mathcal N_r(\mathbf p^\star)\), the local two-sided distance bounds hold there, and \eqref{eq:distance_distortion} holds for its admissible references. If every \(\mathbf p_i\in\mathcal D_c\setminus\mathcal N_r(\mathbf p^\star)\) satisfies
\begin{equation}
d_q(\mathbf p^\star,\mathbf p_i)>Le_g+\varepsilon_d,
\qquad \mathbf p_i\notin\mathcal N_r(\mathbf p^\star),
\label{eq:remote_margin}
\end{equation}
then any minimum-distance reference \(\hat{\mathbf p}_{\mathrm{NN}}\) is local and
\begin{equation}
\|\hat{\mathbf p}_{\mathrm{NN}}-\mathbf p^\star\|
\le (Le_g+2\varepsilon_d)/m.
\label{eq:nearest_bound}
\end{equation}
\end{proposition}
\begin{proof}
The computed distance to the nearest geometric reference is at most \(Le_g+\varepsilon_d\). Equation~\eqref{eq:remote_margin} excludes remote minimizers. Optimality gives
\(d_0(\mathbf p^\star,\hat{\mathbf p}_{\mathrm{NN}})
\le Le_g+2\varepsilon_d\), and the local lower bound gives \eqref{eq:nearest_bound}.
\end{proof}

For the interpolator, weights satisfying \(w_i\ge0\) and \(\sum_iw_i=1\) give the exact inequality
\begin{equation}
\|\mathbf p_{\mathrm{int}}-\mathbf p^\star\|
\le\sum_{i\in\mathcal T_+}w_i\|\mathbf p_i-\mathbf p^\star\|.
\label{eq:convex_bound}
\end{equation}

For the interpolation bound, fix the query \(\mathbf s_o\). Assume a smooth, fixed-support neighborhood containing the true position \(\mathbf p^\star\) and all participating references, where \(|\mathbf c(\mathbf p)^H\mathbf s_o|\) stays bounded away from zero and above the alignment threshold. Define \(\mathbf h_o(\mathbf p)=T_{\mathbf s_o}(\mathbf c(\mathbf p))\) by \eqref{eq:phase_chart} and \(\boldsymbol\psi_o(\mathbf p)=[\operatorname{Re}\mathbf h_o(\mathbf p)^T,\operatorname{Im}\mathbf h_o(\mathbf p)^T]^T\). Set \(\mathbf z_\star=\boldsymbol\psi_o(\mathbf p^\star)\) and \(\mathbf J_\star=D\boldsymbol\psi_o(\mathbf p^\star)\), with differentiation only in position \(\mathbf p\). Write
\(\widetilde{\mathbf z}_i=\mathbf z_\star+\mathbf J_\star\Delta\mathbf p_i+\mathbf r_i\),
where \(\Delta\mathbf p_i=\mathbf p_i-\mathbf p^\star\).
For any nonnegative weights summing to one, if \(\sigma_{\min}(\mathbf J_\star)\ge m_z>0\), then
\begin{equation}
\|\mathbf p_{\mathrm{int}}-\mathbf p^\star\|
\le
\frac{\|\widetilde{\mathbf Z}\mathbf w-\mathbf z_\star\|
+\sum_iw_i\|\mathbf r_i\|}{m_z}.
\label{eq:interpolation_bound}
\end{equation}

Summing the expansions isolates \(\mathbf J_\star(\mathbf p_{\mathrm{int}}-\mathbf p^\star)\); the singular-value and triangle inequalities give the result. For an ideal query, \(\mathbf z_o=\mathbf z_\star\), and the residual term equals \(r_{\mathrm{int}}\max\{\|\mathbf z_\star\|,\epsilon_{\mathrm{mach}}\}\). With perturbed coefficients, add the query error and weighted reference errors measured in this same query-defined chart. The phase map is locally smooth under the nonzero-inner-product condition. The bound also requires stable support and a positive minimum singular value of the Jacobian.

\begin{corollary}[Quantization and near--far separation]
\label{cor:quantization_margin}
Fix the same nonempty near and far index subsets in \eqref{eq:feature_margin} before and after quantization. If \(|d_q-d_0|\le\varepsilon_q\) for every candidate and the unquantized margin satisfies \(\gamma_0>2\varepsilon_q\), then \(\gamma_q>0\): the quantized minimum near score remains smaller than the quantized minimum far score.
\end{corollary}
\begin{proof}
Quantization increases the minimum near score by at most \(\varepsilon_q\) and decreases the minimum far score by at most \(\varepsilon_q\). Thus \(\gamma_q\ge\gamma_0-2\varepsilon_q>0\).
\end{proof}
This condition preserves near--far separation; individual reference identities, Top-\(K_1\) membership, and interpolation weights may still change. With exact scale \(A\), the report remains partly continuous, but amplitude and phase rounding still invalidate a general strict local lower bound. The continuous-map results above apply before rounding.

Localization error depends on map sampling, feature representation, ambiguity, and local nonlinearity. Quantization further perturbs the matching scores. Grid sampling controls \(e_g\) and the available interpolation support. Port truncation changes the feature representation and its local sensitivity. Ambiguity determines whether a competing reference outside the true neighborhood has a smaller distance. Local nonlinearity governs the accuracy of feature interpolation within a chart. These mechanisms interact; they are not independent additive random errors.

For a fixed comparison index set \(\mathcal I\), let \(\mathcal I_{\mathrm n}=\{i\in\mathcal I:\mathbf p_i\in\mathcal N_r(\mathbf p^\star)\}\) and \(\mathcal I_{\mathrm f}=\mathcal I\setminus\mathcal I_{\mathrm n}\). The matching result uses \(\mathcal I=\mathcal C\); a database diagnostic can use its full eligible reference set. Define the margins using identical subsets:
\begin{equation}
\begin{split}
\gamma_t(\mathbf p^\star;r)={}&
\min_{i\in\mathcal I_{\mathrm f}}d_t(\mathbf p^\star,\mathbf p_i)\\
&-\min_{i\in\mathcal I_{\mathrm n}}d_t(\mathbf p^\star,\mathbf p_i),
\quad t\in\{0,q\}.
\end{split}
\label{eq:feature_margin}
\end{equation}

Both subsets must be nonempty for the margin to be defined. A positive margin indicates a strict nearest-reference advantage for the local set, rather than correctness of every retrieved point or the final interpolation.

On a smooth chart, the remainder is
\begin{equation}
\mathbf r(\Delta\mathbf p)=
\bar\Phi(\mathbf p^\star+\Delta\mathbf p)-\bar\Phi(\mathbf p^\star)
-\mathbf J_{\bar\Phi}(\mathbf p^\star)\Delta\mathbf p.
\label{eq:nonlinearity}
\end{equation}

Continuous differentiability gives \(o(\|\Delta\mathbf p\|)\); an \(O(\|\Delta\mathbf p\|^2)\) bound requires a Lipschitz Jacobian. A support switch is not a small second-order remainder of the same chart. This distinction is essential for interpreting the local diagnostic in Section~V.

Reconstruction distortion, margin statistics, and finite-difference diagnostics describe complementary aspects of the observed errors. We use these diagnostics to summarize empirical trends across port budgets.

\section{Numerical Results}
\label{sec:results}
\subsection{Simulation Setup}

The channel maps are generated from DeepMIMO \cite{alkhateeb2019deepmimo} for \texttt{city\_0\_newyork\_3p5} and \texttt{asu\_campus\_3p5}, referred to as City and Arizona State University (ASU). Table~\ref{tab:exp_params} lists the arrays, frequency samples, and localization settings. The first-to-last frequency span is 1.50~MHz, whereas \(K\Delta f=1.56\)~MHz is the DFT sampling-bandwidth parameter. Reference channels and, unless otherwise stated, query estimates are ideal. Fig.~\ref{fig:coefficient_estimation} separately evaluates noisy query coefficients. Fig.~\ref{fig:scene_maps} shows the active layouts.

\begin{figure}[!t]
\centering
\subfloat[ASU campus.]{\includegraphics[width=0.78\columnwidth]{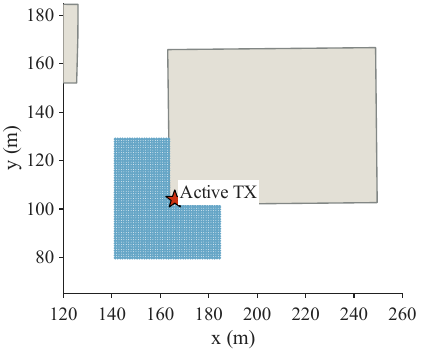}}\\
\subfloat[City.]{\includegraphics[width=0.78\columnwidth]{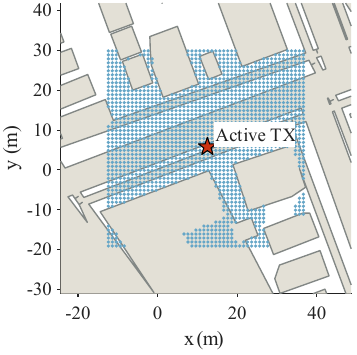}}
\caption{DeepMIMO layouts with building outlines, valid locations, and the active transmitter.}
\label{fig:scene_maps}
\end{figure}

Both maps use one transmitter, identified by source tags \texttt{t001/tx000}, and receiver set \texttt{r000}; no multi-BS observations are combined. The BS coordinates are \((12.4973,5.8259,6)\)~m in City and \((166,104,22)\)~m in ASU. UE heights are fixed at 2~m and 1.5~m, respectively. The synthesis uses half-wavelength spacing, arrays in the coordinate \(xy\)-plane, and no UE rotation. Its response is \(\mathbf a_x\otimes\mathbf a_y\), with the y index varying fastest. For zenith \(\theta\) and azimuth \(\phi\), element \((m_x,m_y)\) has phase \(2\pi(0.5m_x\sin\theta\cos\phi+0.5m_y\sin\theta\sin\phi)\). The \(4\times1\) UE array lies along the x axis. These are the synthesis conventions applied to the supplied path angles; their physical frame alignment remains a dataset-registration assumption. In Fig.~\ref{fig:scene_maps}, beige polygons are top-down projections of the supplied building meshes, blue points are valid locations, and the star marks the active transmitter. Other transmitters are omitted; the view extends beyond the sampled region to show surrounding structures.

Seed-1 coarse priors are fixed across methods and noise realizations. Leave-one-out removes the test reference; sparse maps retain the full test set. Offset directions are uniform on \([0,2\pi)\), with magnitudes uniform in radius on \([3,10]\)~m and \(R=10\)~m by default. The joint prior/search experiment sets both offset magnitude and \(R\) to \(2,4,6,8,10\)~m, placing the true position on the search-circle boundary and potentially creating asymmetric neighborhoods.

All query errors enter
\begin{equation}
\mathrm{RMSE}=\sqrt{\frac1{N_{\mathrm{test}}}
\sum_{n=1}^{N_{\mathrm{test}}}\|\hat{\mathbf p}_n-\mathbf p_n^\star\|^2}.
\label{eq:rmse_metric}
\end{equation}

An empty candidate set or a nonfinite estimate triggers a return to the query's coarse position. These cases are counted as failures and included in RMSE. Returning the matched anchor after rejected interpolation is treated as a valid estimate. Large finite errors are not clipped.

\begin{table}[!t]
\centering
\caption{Evaluation and Localization Settings}
\label{tab:exp_params}
\renewcommand{\arraystretch}{1.10}
\begin{tabularx}{\columnwidth}{@{}lX@{}}
\toprule
Parameter & Setting\\
\midrule
Valid locations & City: 1602; ASU: 1631\\
Carrier / arrays & 3.5 GHz; BS \(8\times8\), UE \(4\times1\)\\
Frequency samples & \(K=26\), \(\Delta f=60\) kHz\\
Frequency span & \((K-1)\Delta f=1.50\) MHz\\
Port scan & \(P=1,2,4,8,16,32\)\\
Map/prior comparisons & City: \(P=8\); ASU: \(P=32\)\\
Nominal map spacing & \(1,2,4,8\) m\\
Default prior / search & Magnitude \(U[3,10]\) m; \(R=10\) m\\
Retrieval / interpolation & \(K_{\mathrm r}=8\); \(\lambda=10^{-3}\); minimum 4 references\\
Acceptance checks & Residual \(\le0.35\); displacement \(\le\max(1.5h_{\mathcal T},1~\mathrm m)\)\\
ADPS-AWKNN & \(K_{\max}=16\), \(r_{\mathrm{th}}=10^{-4}\)\\
Reference maps / prior seed & Ideal / 1\\
Coefficient-error experiment & 30 noise realizations per estimation SNR\\
\bottomrule
\end{tabularx}
\end{table}

The ADPS-AWKNN baseline adapts Hu--Chen spectral fingerprinting \cite{hu2024adps} to a UPA and multiple receive antennas. For each receive antenna \(r\), reshape the full channel into \(\mathbf Z_r\in\mathbb C^{N_t\times K}\). The fingerprint is the normalized vectorization of
\begin{equation}
\mathbf A=
\left(\sum_{r=1}^{N_r}
|(\mathbf D_\theta\mathbf Z_r)\odot(\mathbf Z_r\mathbf T_\tau)|^2
\right)^{1/2}.
\label{eq:adps}
\end{equation}

Here \(\odot\) denotes elementwise multiplication, the square root is elementwise, and
\(\mathbf D_\theta=\mathbf D_x\otimes\mathbf D_y\) follows the channel's antenna ordering. The spatial transform entries are
\([\mathbf D_a]_{m,n}=N_a^{-1/2}e^{-j2\pi d_a\cos(\pi m/N_a)n}\),
with \(m=1,\ldots,N_a\), \(n=0,\ldots,N_a-1\), and element spacing \(d_a\) in wavelengths. The frequency transform has
\([\mathbf T_\tau]_{k,n}=K^{-1/2}e^{-j2\pi(f_k-f_1)n/(K\Delta f)}\),
\(n=1,\ldots,K\). Root-sum-square receive aggregation and the magnitude of the Hadamard product specify the adaptation.

Candidates in the same prior region are ranked by cosine similarity, clipped to \([0,1]\). Let \(s_1\ge\cdots\ge s_{|\mathcal C|}\) denote these scores and \(J_{\max}=\min(K_{\max},|\mathcal C|)\). The adaptive rule searches \(j=2,\ldots,J_{\max}\) for
\((s_1-s_j)/\max(\sum_{i=1}^{|\mathcal C|}s_i,\epsilon)\ge r_{\mathrm{th}}\) and retains \(j-1\) neighbors at the first crossing. The denominator sums over all local candidates, not just the searched prefix; \(\epsilon\) is the smallest positive normalized floating-point value. With no crossing, \(J_{\max}-1\) neighbors are retained, so the default \(K_{\max}=16\) permits at most 15 neighbors. A perfect first match or \(J_{\max}=1\) uses one neighbor. Retained similarities are normalized as position weights, with equal weights if their sum is zero. Both the adaptation and its \(K_{\max}=16\), \(r_{\mathrm{th}}=10^{-4}\) settings are held fixed.

We additionally evaluate a Top-\(M\) sparse ADPS-AWKNN extension of this baseline. For each nonnegative fingerprint \(\mathbf f\), define
\begin{equation}
\mathcal S_M=\operatorname{Top}_M(\mathbf f),\qquad
[\mathbf f_M]_d=f_d\,\mathbf1\{d\in\mathcal S_M\}.
\label{eq:topm_adps}
\end{equation}
The largest values are retained, with ties resolved by ascending dimension index. Map and query fingerprints are truncated independently. Cosine similarity uses their truncated norms; AWKNN parameters, neighbor selection, tie handling, and fallback rules are unchanged. We test \(M\in\{4,8,14,16,32,57,64,128,256,512,1024,1664\}\), with \(M=D=1664\) recovering dense ADPS. The storage comparison uses ideal estimates, unquantized double-precision fingerprints, and identical dense maps, queries, exclusions, and seed-1 priors. Elsewhere, ADPS-AWKNN denotes the dense baseline.

The spatial-smoothing MUSIC (SS-MUSIC) baseline combines direction estimates with received signal strength (RSS), and is denoted SS-MUSIC/RSS. It uses the full channel and spatial smoothing over sixteen overlapping \(5\times5\) BS subarrays. Its signal subspace has fixed rank one. It searches azimuths \(-180^\circ:2^\circ:178^\circ\) and zeniths \(92^\circ:2^\circ:178^\circ\), retains at most 12 local spectral peaks, and intersects their rays with the scene-wide median UE-height plane. Candidates inside the prior disk are scored using prior/RSS/spectrum weights \(1/0.8/0.2\), an RSS scale of 5~dB, and inverse-distance-square RSS interpolation from up to four map references. No true path angle, LOS label, or isolated first path is supplied. If no angular candidate is admissible, the coarse-prior fallback applies.

\subsection{Localization Accuracy and Map Storage}
We compare reference-map thinning and joint changes of the prior and search radius in Fig.~\ref{fig:accuracy}. Nominal spacings of \(1,2,4,8\)~m retain every corresponding index along the sorted unique coordinate axes, with a common test set. Supplementary Table~S1 gives the full map-density results.

\begin{figure*}[!t]
\centering
\includegraphics[width=2\columnwidth]{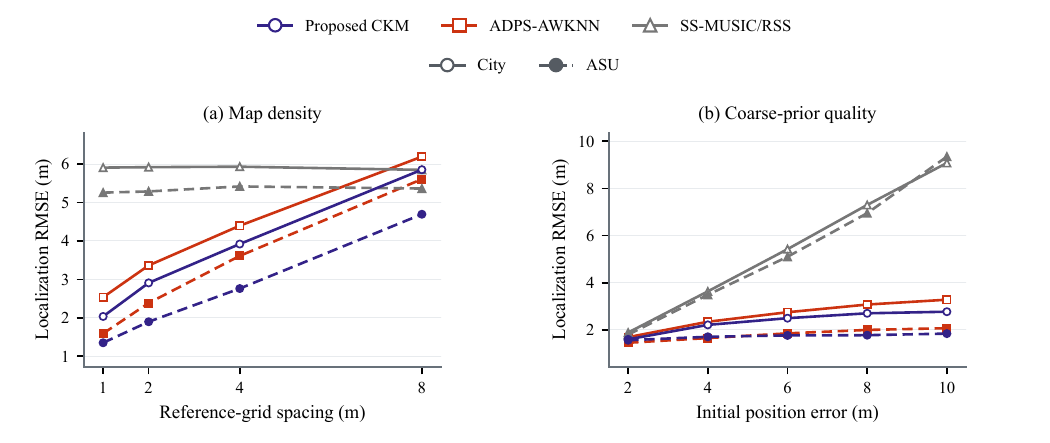}
\caption{Localization RMSE versus (a) map spacing and (b) joint prior-error magnitude and search radius.}
\label{fig:accuracy}
\end{figure*}

The proposed method has lower RMSE than dense ADPS-AWKNN at every tested spacing. Thinning degrades both methods. City RMSE rises by approximately 3.82~m from 1 to 8~m spacing. At 8~m, it is effectively tied with SS-MUSIC/RSS (5.8543 versus 5.8496~m). MUSIC depends less on map density because references support RSS-based candidate selection rather than fingerprint matching.

In Fig.~\ref{fig:accuracy}(b), increasing the offset magnitude and search radius together admits more references while making candidate selection harder. At 10~m, the proposed and ADPS-AWKNN RMSEs are 2.7732 and 3.2824~m in City, and 1.8444 and 2.0743~m in ASU. At the 2~m ASU setting, ADPS-AWKNN is better, 1.4557 versus 1.5644~m. These trends reflect joint changes in prior error and search radius, with the scene-specific port budgets held fixed.

With identical default-prior retrieval and thresholds, Top-1, unaligned NNLS, and aligned NNLS give City RMSEs of 2.2994, 2.2394, and 2.0352~m, and ASU RMSEs of 1.6527, 1.5921, and 1.3500~m. Alignment reduces P90 from 3.6056 to 3.1806~m in City and from 2.2361 to 2.0000~m in ASU. Independently rotating every query and reference in three trials per scene changes positions by less than \(1.3\times10^{-13}\)~m, with unchanged retrieval and acceptance. At the 10~m boundary-prior setting, Top-1 gives RMSEs of 2.5616~m in City and 1.7593~m in ASU, compared with 2.7732~m and 1.8444~m for aligned interpolation. In this setting, the true position lies on the search-circle boundary, which can produce one-sided reference neighborhoods.

At 8~m map spacing, proposed and ADPS-AWKNN failures are 0.7491\% in City and 2.4525\% in ASU, versus zero at 1~m. SS-MUSIC/RSS has substantial fallback, approximately 23.60\% and 7.91\% at this spacing. All RMSE values include these coarse-position returns; Supplementary Fig.~S1 reports failure rates separately.

Fig.~\ref{fig:features} evaluates \(P=1,2,4,8,16,32\). Increasing from one to eight ports reduces RMSE from 2.7462 to 2.0352~m in City and from 2.8851 to 1.4167~m in ASU. Beyond eight ports, City nearly plateaus, while ASU gains another 0.0667~m at 32 ports. All queries produce finite outputs.

\begin{figure}[!t]
\centering
\includegraphics[width=\columnwidth]{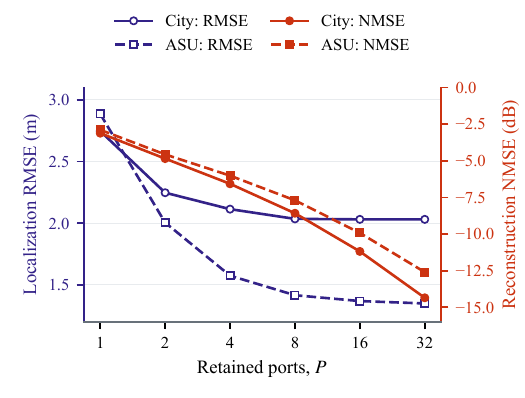}
\caption{Port-budget dependence of localization RMSE (left axis) and reconstruction NMSE (right axis).}
\label{fig:features}
\end{figure}

The reconstruction diagnostic uses \(\hat{\mathbf X}\) synthesized from the retained ports, and aligns it to \(\mathbf X\) only for offline assessment. With \(\mathbf x=\operatorname{vec}(\mathbf X)\) and \(\hat{\mathbf x}=\operatorname{vec}(\hat{\mathbf X})\),
\begin{equation}
\alpha=\frac{\hat{\mathbf x}^H\mathbf x}{\hat{\mathbf x}^H\hat{\mathbf x}},
\qquad
E_{\mathrm{rec}}=\frac{\|\mathbf x-\alpha\hat{\mathbf x}\|^2}{\|\mathbf x\|^2}.
\label{eq:reconstruction}
\end{equation}

We plot \(10\log_{10}(N^{-1}\sum_n E_{\mathrm{rec},n})\): the mean linear NMSE is converted to decibels. It evaluates the effective channel after receive combining. Between one and 32 ports, City improves from \(-3.1203\) to \(-14.3509\)~dB and ASU from \(-2.8916\) to \(-12.6000\)~dB. Reconstruction continues to improve after localization accuracy has largely saturated, so reconstruction fidelity alone does not determine positioning accuracy.

Fig.~\ref{fig:error_correlation} combines three raw diagnostics: linear reconstruction NMSE \(q_1\), ambiguity \(q_2=1-q_+\), and local nonlinearity \(q_3\). Here \(q_+\) is the fraction of all evaluated positions with a defined, positive unquantized margin \(\gamma_0\) in \eqref{eq:feature_margin}, using a 5~m neighborhood; \(q_3\) averages \(\|\Delta\mathbf c-\mathbf J\Delta\mathbf p\|/\max(\|\Delta\mathbf c\|,\epsilon_{\mathrm{mach}})\) over admissible neighborhoods. For the nondegenerate ranges used here,
\begin{equation}
\widetilde q_j=\frac{q_j-a_j}{b_j-a_j},\qquad
I_{\mathrm{diag}}=\sum_{j=1}^{3}\widetilde q_j.
\label{eq:diagnostic_normalization}
\end{equation}

All configurations share the component-wise normalization bounds specified in Supplementary Section~I.

Configurations with a nonfinite component are excluded from the composite-index plot. This applies to City at \(P=32\), whose localization RMSE remains in Fig.~\ref{fig:features}. The remaining five City and six ASU configurations yield Pearson correlations of 0.9943 and 0.9834. Supplementary Section~I specifies the neighborhood selection and finite-difference calculation of \(q_3\). Figure~\ref{fig:error_correlation} relates pre-interpolation feature diagnostics to the final localization RMSE.

\begin{figure}[!t]
\centering
\includegraphics[width=\columnwidth]{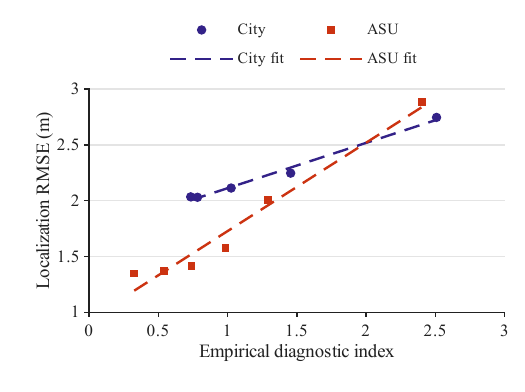}
\caption{Composite diagnostic index versus localization RMSE, with scene-specific linear fits.}
\label{fig:error_correlation}
\end{figure}

For the double-precision comparison, each feedback component contains a complex value and a \texttt{uint16} index, whereas each sparse ADPS component contains a real value and the same index type. The per-entry payloads in bytes are
\begin{equation}
\begin{aligned}
S_{\mathrm{feedback}}(P)&=18P,\\
S_{\mathrm{sparse\ ADPS}}(M)&=10M,\quad M<D,\\
S_{\mathrm{dense\ ADPS}}&=8D=13312.
\end{aligned}
\label{eq:storage}
\end{equation}

At \(M=D\), indices are implicit and the dense format applies. Fixed-size sparse entries retain their indices even when a selected value is zero. The feedback dictionary is shared, and support indices fully specify its outer descriptor. Coordinates, common metadata, and container overhead are excluded equally; temporary computation arrays are not map payloads.

\begin{table}[!t]
\centering
\caption{Per-Entry Storage and All-Sample RMSE}
\label{tab:storage}
\setlength{\tabcolsep}{3.5pt}
\begin{tabular}{@{}llrr@{}}
\toprule
Scene & Method & Bytes & RMSE (m)\\
\midrule
City & Feedback, \(P=8\) &144&2.0352\\
& Sparse ADPS, \(M=14\) &140&2.7724\\
& Dense ADPS &13312&2.5360\\
ASU & Feedback, \(P=32\) &576&1.3500\\
& Sparse ADPS, \(M=57\) &570&1.5811\\
& Dense ADPS &13312&1.5909\\
\bottomrule
\end{tabular}
\end{table}

Fig.~\ref{fig:storage_accuracy} plots all tested budgets. Top-\(M\) truncation substantially reduces ADPS storage: \(M=57\) uses 570 bytes, 95.7\% less than dense ADPS, while slightly improving RMSE in both scenes. Truncation changes cosine similarities and AWKNN neighbor selection, so RMSE need not vary monotonically with storage. Diamonds mark dense ADPS; enlarged markers highlight the 144/140-byte and 576/570-byte comparisons.

At similar storage budgets (Table~\ref{tab:storage}), the proposed method reduces RMSE relative to sparse ADPS by 26.6\% in City (144 versus 140 bytes) and 14.6\% in ASU (576 versus 570 bytes).

At similar accuracy in ASU, feedback with \(P=4\) gives 1.5752~m using 72 bytes, versus 1.5811~m and 570 bytes for sparse ADPS with \(M=57\), a storage reduction of 87.4\%. These results compare the complete localization methods. The proposed method also reuses communication feedback directly, whereas sparse ADPS constructs its spectrum from channel observations.

\begin{figure*}[!t]
\centering
\includegraphics[width=6.831in]{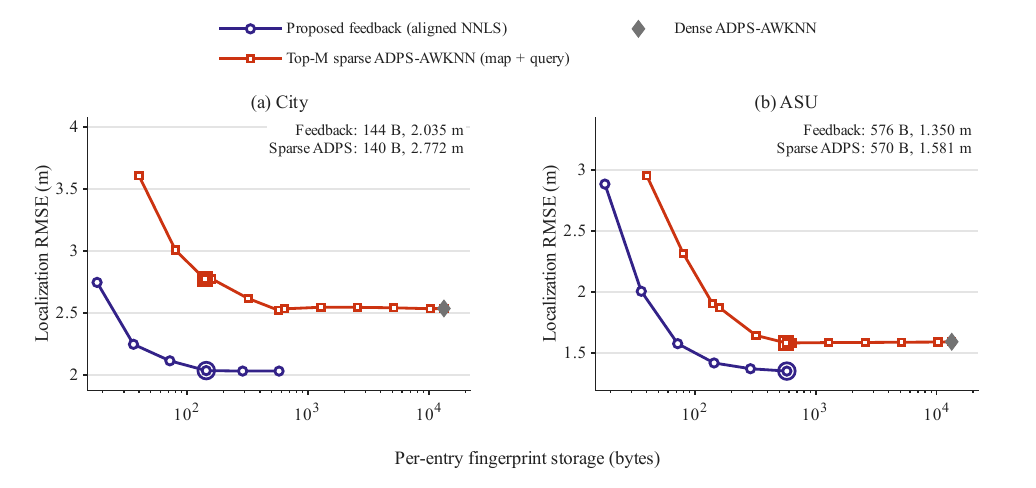}
\caption{Per-entry map storage versus localization RMSE in (a) City and (b) ASU.}
\label{fig:storage_accuracy}
\end{figure*}

\subsection{Quantization and Coefficient Estimation Error}
\label{subsec:quantization}
Both reference fingerprints and query coefficients are independently quantized using the same amplitude--phase quantizer, with localization and coarse priors unchanged. Support indices and the float64 scale \(A\) are retained exactly. Fig.~\ref{fig:quantization_ablation} compares amplitude-only, phase-only, and joint quantization against the floating-point baseline. Supplementary Fig.~S3 reports the query-only control with a floating-point reference map.

\begin{figure*}[!t]
\centering
\includegraphics[width=6.831in]{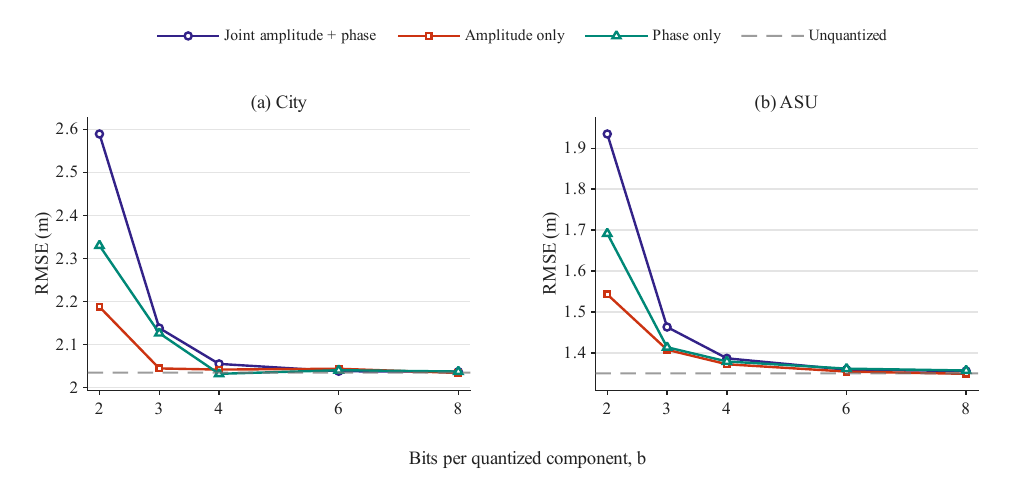}
\caption{Amplitude--phase quantization of map and query coefficients in (a) City and (b) ASU.}
\label{fig:quantization_ablation}
\end{figure*}

The quantizer is defined in \eqref{eq:amplitude_phase_normalization}--\eqref{eq:amplitude_phase_quantizer}. City uses eight ports and ASU uses 32, with the same seed-1 queries as the floating-point comparison. Joint quantization assigns \(b\) bits to each amplitude and phase component; component-only ablations retain the other component in floating point.

With two bits per component, joint quantization increases RMSE by 27.23\% in City and 43.31\% in ASU. At eight bits, RMSEs are 2.0374 and 1.3557~m, only 0.0022 and 0.0057~m above the floating-point baselines. These settings have zero failed outputs. Small changes across precision settings need not be monotonic.

Moderate coefficient precision therefore retains nearly the floating-point accuracy in these maps. A joint report's modeled payload is \(P(11+b_a+b_\phi)+64\) bits, excluding protocol metadata. The 11-bit index accounting differs from the byte-aligned 16-bit index format in Table~\ref{tab:storage}.

The quantization experiment assumes ideal coefficient estimation. To isolate estimation error, we now use unquantized reports and an ideal reference map, with support and receive combining fixed. For each query, we generate
\begin{equation}
\widetilde{\mathbf c}=\mathbf c+\mathbf e,\qquad
\mathbf e\mid\mathbf c\sim\mathcal{CN}(\mathbf0,\sigma_c^2\mathbf I_P),
\label{eq:feedback_estimation_error}
\end{equation}
and define the feedback-coefficient estimation signal-to-noise ratio (SNR) by
\begin{equation}
\rho_{\mathrm{fb}}=\frac{\|\mathbf c\|_2^2}{P\sigma_c^2},\qquad
\sigma_c^2=\frac{\|\mathbf c\|_2^2}{P\,10^{s/10}},
\label{eq:feedback_estimation_snr}
\end{equation}
where \(s\) is the coefficient estimation SNR in decibels, rather than wireless-link SNR at fixed receiver noise power. Localization receives only the noisy query coefficients. City uses \(P=8\) and ASU uses \(P=32\); maps, exclusions, coarse priors, and localization settings are unchanged. We run 30 independent noise realizations at each \(s\in\{0,5,10,15,20,30,40\}\)~dB, scaling the same standard Gaussian vector across SNRs within each realization. RMSE pools all queries and realizations, including coarse fallbacks. The 95\% bootstrap intervals use 10,000 resamples of the 30 realization-level MSEs, followed by a square root. They quantify noise variability for these fixed maps, queries, and priors.

\begin{figure}[!t]
\centering
\includegraphics[width=3.465in]{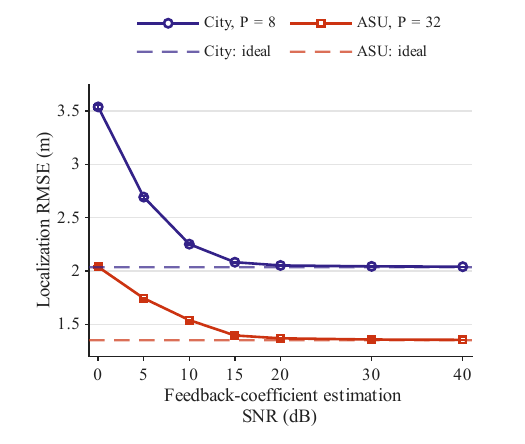}
\caption{Localization RMSE versus feedback-coefficient estimation SNR, with 95\% bootstrap intervals.}
\label{fig:coefficient_estimation}
\end{figure}

In Fig.~\ref{fig:coefficient_estimation}, RMSE decreases with estimation SNR. At 0~dB, City and ASU give 3.5385 and 2.0382~m, compared with ideal-estimation values of 2.0352 and 1.3500~m. At 30~dB, both are within 0.52\% of their baselines. No query falls back to the coarse prior, but interpolation acceptance falls from 72.28\% and 84.86\% without error to approximately 0.052\% and 0\% at 0~dB. At low SNR, most outputs are therefore matched anchors. Solid curves show pooled RMSE, bands and error bars show the bootstrap intervals, and dashed lines mark ideal-estimation baselines.

\FloatBarrier
\section{Conclusion}
Beam feedback coefficients provide compact CKM fingerprints for coarse-prior-guided localization. Phase-invariant matching and query-aligned interpolation reuse communication reports without reconstructing full channels. The analysis characterizes local identifiability and finite-map errors. Compared with dense ADPS-AWKNN, the proposed method uses substantially less storage and achieves lower RMSE in both scenes. It also achieves lower RMSE than Top-\(M\) sparse ADPS-AWKNN at similar storage budgets. Eight-bit amplitude and phase quantization remains close to floating-point performance; coefficient estimation at 30~dB increases RMSE by less than 0.52\% in both scenes. Future work should examine support-selection and receive-combining errors, protocol-level reporting, and the role of neighborhood geometry in interpolation.

\bibliographystyle{IEEEtran}
\bibliography{ckm_refs}
\end{document}